\documentclass[a4paper]{article}

\usepackage{amsmath,amssymb,graphicx,xspace}
\usepackage[margin=1in]{geometry}

\usepackage{amsthm}
\theoremstyle{plain}
\newtheorem{theorem}{Theorem}
\newtheorem{lemma}{Lemma}
\newtheorem{corollary}{Corollary}
\theoremstyle{definition}
 
\newtheorem{definition}{Definition}

\def\eps{{\varepsilon}}

\newcommand{\emphi}[1]{\emph{#1}}
\newcommand{\VC}{\ensuremath{\mathsf{VC}}\xspace}
\newcommand{\etal}{\textit{et~al.}\xspace}
\newcommand{\Komlos}{K\'o{}m{}l{}o{}s\xspace}

\newcommand{\R}{\mathcal{R}}
\newcommand{\E}{\mathcal{E}}

\newcommand{\cardin}[1]{\lvert {#1} \rvert}
\newcommand{\pth}[1]{\!\left({#1}\right)}

\title{%
  On interference among moving sensors and related problems\thanks{A preliminary version of this paper appeared in the proceedings of the {\em European Symposium on Algorithms} (ESA 2016)~\cite{ckkrrs-okrsaa-15}.
M.~J.~K.~was partially supported by grant 1884/16 from the Israel Science Foundation.
M.~K.~was partially supported by MEXT KAKENHI Nos.~12H00855, 15H02665, and 17K12635. A.~v.~R. and M.~R. were supported by JST ERATO Grant Number JPMJER1305, Japan. S.~S.~was partially supported by Grant 1136/12 from the Israel Science Foundation and by the Swiss National Science Foundation Grants 200020144531 and 200021-137574.}
}

\author{%
  Jean-Lou~De~Carufel%
  \thanks{University of Ottawa, Ottawa, Canada. 
          \texttt{jdecaruf@uottawa.ca}}\,
  \and
  Matthew~J.~Katz,%
  \thanks{Ben-Gurion University of the Negev, Beer-Sheva, Israel.
          \texttt{matya@cs.bgu.ac.il,shakhar@math.bgu.ac.il}}\,
  \and
  Matias~Korman,%
  \thanks{Tohoku University, Sendai, Japan. 
          \texttt{mati@dais.is.tohoku.ac.jp}}\,
  \and
  Andr\'e~van~Renssen,%
  \thanks{National Institute of Informatics, Tokyo, Japan. 
          \texttt{\{andre,marcel\}@nii.ac.jp}}\, $^{,}$\thanks{JST, ERATO, Kawarabayashi Large Graph Project}\,
  \and
  Marcel~Roeloffzen,\footnotemark[5]\, $^{,}$\footnotemark[6]\,
  \and
  Shakhar~Smorodinsky\footnotemark[2]\, $^{,}$\thanks{\'{E}cole Polytechnique F\'{e}d\'{e}rale de Lausanne, Lausanne, Switzerland.}
}
\date{}

\begin{document}
\maketitle

\begin{abstract}
We show that for any set of $n$ points moving along ``simple'' trajectories (i.e., each coordinate is described with a polynomial of bounded degree) in $\Re^d$ and any parameter $2 \le k \le n$, one can select a fixed non-empty subset of the points of size $O(k \log k)$, such that the Voronoi diagram of this subset is ``balanced'' at any given time (i.e., it contains $O(n/k)$ points per cell). We also show that the bound $O(k \log k)$ is near optimal even for the one dimensional case in which points move linearly in time. As applications, we show that one can assign communication radii to the sensors of a network of $n$ moving sensors so that at any given time their interference is $O(\sqrt{n\log n})$. We also show some results in kinetic approximate range counting and kinetic discrepancy. In order to obtain these results, we extend well-known results from $\eps$-net theory to kinetic environments. 
\end{abstract}

\section{Introduction}
We consider the following kinetic facility location problem: given $n$ clients (i.e., points) that are moving in $\Re^d$ along simple trajectories and a parameter $k\leq n$, we wish to select few of them to become facilities to serve the remaining clients. We follow the usual assumption that at any instant of time a client is served by its nearest facility. Our aim is to select the facilities so that none serves too many customers. Specifically, we wish to maintain the invariant that at any given time the number of clients served by each of the chosen facilities is bounded by $n/k$. 

The pigeon-hole principle directly implies that we cannot select fewer than $k$ facilities. 
Our main result is that a subset of size $O(k \log k)$ will suffice. We also show that one cannot improve this bound to $O(k)$, even for $d=1$. As an application, we show how to construct a communication graph among a set of $n$ moving sensors
such that at any given time, the interference of the communication graph is bounded by $O(\sqrt{n \log n})$ (and its hop-diameter is three). Intuitively speaking, the interference of a sensor is the in-degree (i.e., the number of sensors that can communicate to him directly, see more details in Section~\ref{sec-applis-interference}). This bound is near optimal as already, in the static case, there are examples in which any communication graph has $\Omega(\sqrt{n})$ interference~\cite{HT-interference}.

In order to obtain our results we use the machinery of geometric hypergraphs and the theory of \VC-dimension and $\eps$-nets.
By a geometric hypergraph (also called a range-space) we mean the following: suppose we are given a finite set $P$ of points in $\Re^d$ and a family of simple geometric regions, such as the family of all halfspaces in $\Re^d$. Then we consider the combinatorial structure of the set system $(P,\{h\cap P\})$
where $h$ is any halfspace. A key property of such hypergraphs is bounded \VC-dimension (see Section~\ref{sec-prelim} for exact definitions).
In this paper we study a more general structure by allowing the underlying set of points to move along some ``reasonable'' trajectories (i.e., the coordinates of each point can be described with a polynomial function of bounded degree).
Even though the static case is well-known, little research has been done for the case in which the points move.
We show that those more complex hypergraphs, defined as the union of all hypergraphs obtained at all possible times, still have a bounded \VC-dimension. 

In addition to the above mentioned applications, we believe that the bounded \VC-dimension of such hypergraphs is of independent interest and to the best of our knowledge has not been observed before.
We hope that this paper will have many follow-up applications, since bounded \VC-dimension has applications in many other areas of mathematics and computer science.

The paper is organized as follows: in Section~\ref{sec-prelim} we introduce several key concepts as well as review known results that hold for static range spaces. In Section~\ref{sec-kinet} we extend these results to the kinetic case. In Section~\ref{sec-applis-balanced-vor} we prove our main result concerning  Voronoi diagrams for moving points. The interference problem mentioned above is studied in Section~\ref{sec-applis-interference}. In Section~\ref{sec-other} we present two additional applications that follow from known results and the newly introduced kinetic $\eps$-net machinery. We make a few final remarks in Section~\ref{sec-conclusions}.

\section{Preliminaries and Previous Work}\label{sec-prelim}
A hypergraph $H=(V,\E)$ is a pair of sets such that $\E \subseteq
2^V$ (where $2^V$ denotes the {\em power set} containing all subsets of $V$). A geometric hypergraph is one that can be realized in a
geometric way. For example, consider the hypergraph $H = (V,\E)$,
where $V$ is a finite subset of $\Re^d$ and $\E$ consists of all
subsets of $V$ that can be cut-off from $V$ by intersecting it
with a shape belonging to some family of ``nice'' geometric shapes,
such as the family of all halfspaces.
The elements of $V$ are called {\em
vertices}, and the elements of $\E$ are called {\em hyperedges}.
For a subset
$V' \subseteq V$, the hypergraph $H[V'] = (V',\{V' \cap S \colon S \in
\E\})$ is the {\em sub-hypergraph} induced by $V'$.

We consider the following families of geometric hypergraphs:
Let $P$ be a set of points in $\Re^2$ (or, in general, in $\Re^d$) and let $\R$ be a family of regions in the same space. We refer to the hypergraph $H=(P,\{ P \cap r \colon r \in \R\})$ as the hypergraph induced by $P$ with respect to $\R$.
When $\R$ is clear from the context, we sometimes refer to it as
the hypergraph induced by $P$. In the literature, hypergraphs
that are induced by points with respect to geometric regions of
some specific kind are also referred to as {\em range spaces}.
We sometimes abuse the notation and write $(P,\R)$, instead of $H=(P,E)$, where $E = \{P \cap r \colon r \in \R\}$.

\subsection*{$\eps$-nets and VC-dimension}
A subset $T \subset V$ is called a \emphi{transversal} (or a \emphi{hitting set}) of a
hypergraph $H=(V,\E)$, if it intersects all sets of $\E$. The
\emphi{transversal number} of $H$, denoted by $\tau(H)$, is the
smallest possible cardinality of a transversal of $H$. The
fundamental notion of a transversal of a hypergraph is central in
many areas of combinatorics and its relatives. In computational
geometry, there is a particular interest in transversals, since
many geometric problems can be rephrased as questions on the
transversal number of certain hypergraphs~\cite{MATOUSEK-disc}. An important special
case arises when we are interested in finding a small size set $N
\subset V$ that intersects all ``relatively large'' sets of
$\E$.  This is captured in the notion of an $\eps$-net for a
hypergraph:
\begin{definition}[$\eps$-net]
    Let $H=(V,\E)$ be a hypergraph with $V$ finite. Let $\eps \in
    [0,1]$ be a real number. A set $N \subseteq V$ (not necessarily in
    $\E$) is called an \emphi{$\eps$-net} for $H$ if for every hyperedge $S
    \in \E$ with $|S| \geq \eps|V|$ we have $S \cap N \neq
    \emptyset$.\footnote{An analogous definition applies when $V$ is not necessarily finite and
$H$ is endowed with a probability measure.}
\end{definition}

The well-known result of
Haussler and Welzl \cite{HW-eps-net} provides a combinatorial condition
on hypergraphs that guarantees the existence of small $\eps$-nets
(see below). This requires the following well-studied notion of the
Vapnik-Chervonenkis dimension \cite{VC71}:

\begin{definition}[\VC-dimension]
Let $H=(V,\E)$ be a hypergraph. A subset $X \subset V$ (not
necessarily in $\E$) is said to be \emphi{shattered} by $H$ if
$\{X\cap S\colon  S \in \E\}=2^X$. The \emphi{Vapnik-Chervonenkis
dimension}, also denoted the \emphi{\VC-dimension} of $H$, is the
maximum size of a subset of $V$ shattered by $H$.
\end{definition}

\subsection*{Relation between $\eps$-nets and the \VC-dimension}

Haussler and Welzl~\cite{HW-eps-net} proved the following fundamental
theorem regarding the existence of small $\eps$-nets for
hypergraphs with small \VC-dimension.

\begin{theorem}[$\eps$-net theorem]\label{theo_epsnet}
    Let $H=(V,\E)$ be a hypergraph with \VC-dimension $d$. For
    every $\eps \in (0,1)$, there
    exists an $\eps$-net $N \subset V$ with cardinality at most
    $\displaystyle O\pth{  \frac{d}{\eps}\log\frac{1}{\eps} }$.
\end{theorem}

In fact, it can be shown that a random sample of vertices of size $O(\frac{d}{\eps}\log\frac{1}{\eps})$ is an
$\eps$-net for $H$ with a positive constant probability (see~\cite{ConstructionEpsilonNets} for details on how to compute such nets).

Many hypergraphs studied in computational geometry and learning theory
have a ``small'' \VC-dimension,
where by ``small'' we mean a constant independent of
the number of vertices of the underlying hypergraph.
It is known that whenever range spaces are defined through semi-algebraic sets of
constant description complexity (i.e., sets defined as a Boolean combination of
a constant number of polynomial equations and inequalities of constant maximum degree),
the resulting hypergraph has finite \VC-dimension.
Halfspaces, balls, boxes, etc. are examples of such sets; see, e.g.,~\cite{MATOUSEK,PA95} for more details.

Thus, by Theorem~\ref{theo_epsnet}, these hypergraphs admit ``small'' size $\eps$-nets.
\Komlos \etal \cite{KPW} proved that the bound
$O(\frac{d}{\eps}\log\frac{1}{\eps})$ on the size of an $\eps$-net
for hypergraphs with \VC-dimension $d$ is best possible. Namely,
for a constant $d$, they construct a hypergraph $H$ with
\VC-dimension $d$ such that any $\eps$-net for $H$ must have
size of at least $\Omega(\frac{1}{\eps}\log\frac{1}{\eps})$. Recently, several breakthrough results provided better lower and upper bounds on the size of $\eps$-nets in several special cases \cite{Alon-nets,AES09,PachT11}.

\section{Kinetic hypergraphs}\label{sec-kinet}
We start by extending the concept of geometric hypergraphs to the kinetic model. Let $P = \{p_1,\ldots,p_n\}$ denote a set of $n$ moving points in $\Re^d$, where each point is moving along some ``simple'' trajectory. That is, each $p_i$ is a function $p_i: [0,\infty) \rightarrow \Re^d$ of the form $p_i(t)=(x^i_1(t),\ldots,x^i_d(t))$,
where $x^i_j(t)$ is a univariate polynomial ($1\leq j \leq d$). For a given real number $t \geq 0$ and a subset $P' \subset P$, we denote by $P'(t)$ the fixed set of points $\{p(t) \colon  p \in P'\}$.

Let $\cal R$ be a (not necessarily finite) family of ranges; for example, the family of all halfspaces in $\Re^d$. We define the {\em kinetic hypergraph} induced by $\cal R$:
\begin{definition}[kinetic hypergraph]
Let $P$ be a set of moving points in $\Re^d$ and let $\cal R$ be a family of ranges.
Let $(P, \E)$ denote the hypergraph
where $\E$ consists of all subsets $P' \subseteq P$ for which there exists a time $t$ and a range $r \in \cal R$ such that
$P'(t) = P(t) \cap r$. We call $(P,\E)$ the {\em kinetic hypergraph} induced by $\cal R$.
\end{definition}

As in the static case we abuse the notation and denote the kinetic hypergraph by $(P,\cal R)$. In order to apply our techniques, we need the following ``bounded description complexity'' assumption concerning the movement of the points of $P$. We say that a point $p_i= p_i(t)=(x^i_1(t),\ldots,x^i_d(t))\in P$ moves with {\em description complexity} $s>0$ if for each $1 \leq j \leq d$,
the univariate polynomial $x^i_j(t)$ has degree at most $s$. In the remainder of this paper, we assume that $P(0)$ is in ``general position''. That is, at time $t=0$ no $d + 1$ points of $P(0)$ are on a common hyperplane. This assumption can be removed through usual symbolic perturbation techniques.

\subsection{VC-Dimension of kinetic hypergraphs}
In this section we prove that for many of the static range spaces that have small \VC-dimension, their kinetic counterparts also have small \VC-dimension. We start with the family ${\cal H}_d$ of all halfspaces in $\Re^d$.

\begin{theorem}\label{kinetic-halfspaces}
Let $P \subset \Re^d$ be a set of moving points with bounded description complexity $s$.  Then, the kinetic-range space $(P, {\cal H}_d)$ has  \VC-dimension bounded by $O(d\log d+\log s\log \log s)$.
\end{theorem}

To prove Theorem~\ref{kinetic-halfspaces}, we need the following known definition and lemma (see, e.g., \cite{MATOUSEK}).
The \emph{primal shatter function} of a hypergraph $H=(V,\E)$ denoted by $\pi_H$ is a function:
$$
 \pi_H: \{1,\ldots,|V|\} \rightarrow \mathbb N
$$
defined by $\pi_H(i) = max_{V' \subseteq V, \cardin{V'}=i} |H[V']|$, where $\cardin{H[V']}$ denotes the number of hyperedges in the sub-hypergraph $H[V']$.

\begin{lemma}\label{shattered-lemma}
Let $H=(V,\E)$ be a hypergraph whose primal shatter function $\pi_H$ satisfies
$\pi_H(m) =  O(m^c)$ for some constant $c\geq 2$. Then the \VC-dimension of $H$ is $O(c\log c)$. 
\end{lemma}

We provide a sketch of the proof of Lemma~\ref{shattered-lemma} for the sake of completeness.
\begin{proof}
Let $d$ denote the \VC-dimension of $H$, and let $V' \subseteq V$ be a shattered subset of cardinality $d$.
On one hand it means that the number of possible subsets of $V'$ that can be realized as the intersection of $V'$ and a hyperedge in $\E$ is $2^d$. On the other hand, by our assumption on $\pi_H$, for a subset of size $d$, there can be at most $Ad^c$ hyperedges in the sub-hypergraph induced by it, for some appropriate constant $A$. In other words
we have $2^d \leq \pi_H(d) \leq Ad^c$. This implies that $d = O(c \log c)$. Indeed, for any $d \geq 10Ac \log c$, the above inequality does not hold, which would give a contradiction. This completes the proof of the lemma.
\end{proof}

\begin{proof}[Proof of Theorem~\ref{kinetic-halfspaces}]
By Lemma~\ref{shattered-lemma} it suffices to bound the primal shatter function $\pi_{{\cal H}_d}(m)$ by a polynomial of constant degree.
It is a well known fact that the number of combinatorially distinct half-spaces determined by $n$ (static) points in $\Re^d$ is $O(n^d)$. This can be easily seen by charging hyperplanes to $d$-tuples of points (using rotations and translations) and observing that each tuple can be charged at most a constant (depending on the dimension $d$) number of times. 
Thus, at any given time, the number of hyperedges is bounded by $O(n^d)$.
Next, note that as $t$ varies,
a combinatorial change in the hypergraph $(P(t),\cal R)$ can occur only when $d+1$ points $p_1(t),\ldots,p_{d+1}(t)$ become affinely dependent. Indeed, a hyperedge is defined by a hyperplane that contains $d$ points of $P(t)$, and that hyperedge changes when an additional point of $P(t)$ crosses the hyperplane (and thus $d+1$ points become affinely dependent).
This happens if and only if the following determinant condition holds:

\begin{align}
\begin{vmatrix}
x^1_1(t) & x^1_2(t) & \cdots & x^1_d(t) & 1 \\
x^2_1(t) & x^2_2(t) & \cdots & x^2_d(t) & 1 \\
\vdots & \vdots & \ddots & \vdots & \vdots \\
x^{d+1}_1(t) & x^{d+1}_2(t) & \cdots & x^{d+1}_d(t) & 1\\
\end{vmatrix}
=0
\end{align} where $x^j_i(t)$ denotes the $i$'th coordinate of $p_j(t)$.
The left side of the equation is a univariate polynomial of degree at most $d s$. By our general position assumption this polynomial is not identically zero and thus can have at most $d s$ solutions.

That is, a tuple of $d+1$ points of $P(t)$ generates at most $ds$ events.
Hence, the total number of such events is bounded by $ds{n \choose d+1}\leq dsn^{d+1}\leq n^{d+1+\log(ds)}$. Between any two events we have a fixed set of at most $n^d$ distinct hyperedges, thus we can have $n^{2d+1+\log (ds)}$ distinct hyperedges along all instants of time.

Since each hyperedge is defined by the points on its boundary, this property is hereditary. That is, for any subset $P' \subseteq P$ the hypergraph $H[P']$ has at most $\cardin{P'}^{2d+1+\log (ds)}\leq \cardin{P'}^{2(d+\log (ds))}$ hyperedges. Thus, the shatter function satisfies $\pi_H(m) = O(m^{3(d+\log s)})$. Then by Lemma~\ref{shattered-lemma}, $(P,{\cal H}_d)$ has \VC-dimension at most $O(d\log d+\log s\log \log s)$ as claimed. 
\end{proof}

\paragraph{Remark} For our purposes, we assume that both $d$ and $s$ are fixed constants, which in particular implies that the VC-dimension is a constant. However, we note that the proof shows that the dependence on the curve complexity $s$ is much softer than the dependence on the dimension $d$. For instance, $s$ could be as large as $2^d$ and still not asymptotically affect the VC-dimension bound\footnote{We thank the anonymous referee that pointed this out to us}.

Theorem~\ref{kinetic-halfspaces} can be further generalized to arbitrary ranges with so-called bounded description complexity as defined below:
\begin{theorem}\label{constant-description}
Let  $\cal R$ be a collection of semi-algebraic subsets of $\Re^d$, each of which can be expressed as a Boolean combination of a constant number of polynomial equations and inequalities of maximum degree $c$ (for some constant $c$).
Let $P$ be a set of moving points in $\Re^d$ with bounded description complexity. Then the kinetic range-space $(P,\cal R)$ has bounded \VC-dimension.
\end{theorem}
\begin{proof}
The proof combines Lemma~\ref{shattered-lemma} with Theorem~\ref{kinetic-halfspaces} and the so-called Veronese lifting map from Algebraic Geometry. We omit the details as it is very similar to the proof for the static case. See, e.g., \cite{MATOUSEK}.
\end{proof}

\section{Balanced Voronoi cells for moving points}\label{sec-applis-balanced-vor}
In this section we tackle the facility location problem for a set of moving clients, where the goal is to ensure a balanced division of the load among the facilities at any instance of time.
Given a set $P$ of moving points or {\em clients} in $\Re^d$, 
locate a small number of the points to serve as {\em facilities} so that at every instance of time no facility is serving more than $n/k$ clients. We make the usual assumption that each client goes to its nearest facility. In the following we show how to obtain an almost optimal balancing (up to a $\log{k}$ factor), even under the restriction that facilities may be located only at points of $P$.

\begin{theorem}
\label{sparse-cell}
Let $P = \{p_1,\ldots,p_n\}$ be any set of $n$ moving points in $\Re^d$ with bounded description complexity. For any integer $2 \leq k \leq n$, there exists a subset $N \subset P$ of cardinality $O(k \log k)$, such that for any finite point set $S \subset \Re^d$, and for any time $t \ge 0$, each cell of the Voronoi diagram $\mbox{Vor}\:(N(t) \cup S)$ contains at most $O(n/k)$ points of $P(t)$.
\end{theorem}

Before proceeding with the proof of Theorem~\ref{balanced-assignment} we need the following result. An {\em infinite cone} with apex $a \in \Re^d$ and angle $\theta \in \Re$ is defined as 
the set:
$$
\{x \in \Re^d \colon (x-a)\cdot (b-a) \geq \|x-a\| \cos (\theta/2) \} \ ,
$$
where $``\| \|"$ denotes the Euclidean norm, $``\cdot "$ denotes the dot product and $b$ is a vector such that $\| b-a\|=1$ (intuitively speaking, it contains all halflines emanating from $a$ that form an angle of at most $\theta/2$ with $b$). A {\em bounded cone} is the intersection of an infinite cone with a ball centered at its apex.

\begin{lemma}\label{cones-vc}
Let $P$ be a set of moving points in $\Re^d$ with bounded description complexity $s$, and let $\cal R$ be the family of all bounded cones. 
The kinetic hypergraph $(P,\cal R)$ has bounded \VC-dimension.
\end{lemma}
\begin{proof}
As shown above, the boundary surface of an infinite cone is a quadric (i.e., a polynomial of degree $2$). In particular, the ranges of $\R$ can be expressed as semi-algebraic sets of constant description complexity. Thus, by Theorem~\ref{constant-description} the hypergraph $(P,\R)$ has constant \VC-dimension as claimed.
\end{proof}

\begin{proof}[Proof of Theorem~\ref{sparse-cell}]
Let $\cal W$ be the family of all bounded cones in $\Re^d$.
Let $H=(P,\cal W)$ be the corresponding kinetic hypergraph.
By Lemma~\ref{cones-vc}, $H$ has constant \VC-dimension.

We fix $\eps = \frac{1}{k}$ and let $N \subset P$ be an $\eps$-net for $H$ of size $O(k\log k)$ (refer to Theorem~\ref{theo_epsnet}). We show that $N \cup S$ satisfies the desired property. That is, for any time $t \ge 0$ and point $q \in N \cup S$, the Voronoi cell of $q(t)$ in the Voronoi diagram $\text{Vor}\:(N(t) \cup S)$
contains at most $O(n/k)$ points of $P(t)$.
Let $C_d$ be the minimum number of sixty-degree cones that are needed to cover  the unit sphere ${\cal S}^{d-1}$. Using packing arguments it is easily seen that $C_d$ is a constant that depends only on $d$; see, e.g., \cite{Boroczky}.

Assume to the contrary that the Voronoi cell of $q(t)$ contains a subset $P'(t) \subset P(t)$ of more than $C_d n/k$ points. By definition, each of the points in $P'(t)$ is closer to $q(t)$ than to any other point in $N(t) \cup S$.
By the pigeonhole principle,
there is an infinite sixty-degree cone $W$
which has $q(t)$ as its apex
and that contains at least $n/k+1$ of the points of $P'(t)$.
Sort the points of $P'(t)\cap W$ in increasing distance from $q(t)$; let $p_1(t),\ldots,p_j(t)$ be the obtained order (note that by assumption, we have $j \ge n/k+1$). Slightly perturb the cone $W$ and bound it to obtain a bounded cone $W'$ that contains the points $p_1(t),\ldots,p_j(t)$ but does not contain $q(t)$ (or any other point of $P(t)$). This can always be done by usual symbolic perturbation tricks~\cite{em-sstcdc-90}.
Since $N$ is an $\eps$-net with respect to bounded cones, $W'$ must contain a point $q'(t) \in N(t)$ (other than $q(t)$).

Since any point in $P(t) \cap W'$ also belongs to $W$, which is a cone of sixty degrees, any point $p(t) \in P(t) \cap W'$ for which $d(p(t),q(t)) \ge  d(q'(t),q(t))$ must be closer to $q'(t)$ than to $q(t)$ (the apex of the cone). In particular, $p_{j}(t)$ satisfies this inequality and thus belongs to the Voronoi cell of $q'(t)$ (and not of $q(t)$), which is a contradiction. 
\end{proof}

In the remainder of this paper, we use the following corollary of Theorem~\ref{sparse-cell}, with $S=\emptyset$. 

\begin{corollary}\label{balanced-assignment}
Let $P = \{p_1,\ldots,p_n\}$ be any set of $n$ moving points in $\Re^d$ with bounded description complexity. For any integer $2 \leq k \leq n$, there exists a subset $N \subset P$ of cardinality $O(k \log k)$, such that for any time $t \ge 0$, each cell of the Voronoi diagram $\mbox{Vor}\:(N(t))$ contains at most $O(n/k)$ points of $P(t)$.
\end{corollary}

\paragraph*{Remark}
We note that the bound of $O(k \log k)$ in Corollary~\ref{balanced-assignment} is near optimal. Clearly, if there are only $o(k)$ points in $N$ then by the pigeonhole principle one of the Voronoi cells must contain $\omega(n/k)$ points of $P$. We also note that reducing the size of the set $N$ seems to be out of reach and maybe impossible, even for the one dimensional case where the points move with constant speed. This follows from a recent lower-bound construction of Alon~\cite{Alon-nets} for $\eps$-nets for static hypergraphs consisting of points with respect to strips in the plane.

\begin{corollary}
\label{lower}
Let $P = \{p_1,\ldots,p_n\}$ be any set of $n$ moving points in $\Re$ moving linearly. There does not exist a subset $N \subset P$ of cardinality $O(k)$, such that for any time $t \ge 0$, each cell of the Voronoi diagram $\mbox{Vor}\:(N(t))$ contains at most $O(n/k)$ points of $P(t)$.
\end{corollary}
\begin{proof}
Indeed, for the sake of contradiction, assume that each point $p\in P$ is described with a linear equation of the form $p(t) = at + b$ (i.e., a line) and there exists a subset $N\subset P$ such that for any $t>0$ and $q\in N$, the Voronoi cell of $q(t)$ contains at most $n/k$ points of $P(t)$. In particular, this implies that there are at most $2n/k$ points of $P(t)$ between any pair of consecutive points of $N(t)$. If we view the moving points in $\Re$ as lines in $\Re^2$, this is equivalent to choosing a subset of the lines with the property that any vertical segment (i.e., a range of the form $t_0 \times [c,d]$ for constants $t_0>0$, $c,d\in \mathbb{R}$) that intersects more than $2n/k$ of the above lines will also intersect one of the chosen lines. By standard point-line duality in two dimensions, this is equivalent to the problem of finding an $\eps=\frac{2}{k}$-net for points with respect to strips in the plane, which still remains an open problem. Recently, Alon~\cite{Alon-nets} gave a construction showing that such hypergraphs cannot have linear (in $\frac{1}{\eps}$) size $\eps$-nets. Since their problem can be reduced to ours, the same lower bound holds for our problem.
\end{proof}

\section{Low interference for moving transmitters}\label{sec-applis-interference}
Here we show how to tackle the problem of minimizing interference among a set of wireless moving transmitters while keeping the number of topological changes of the underlying network subquadratic.
In the following we define the concept of (receiver-based) {\em interference} of a set of ad-hoc sensors \cite{rwz-2009-amiwasn} (see Figure~\ref{fig_interf}).

\begin{figure}[ht]
\centering
\includegraphics{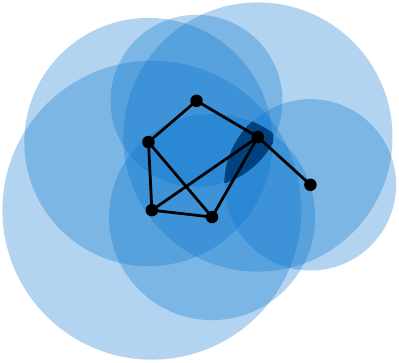}
\caption{Given a set of fixed points in $\Re^2$ and their power assignments represented by disks, the interference is the deepest point in the arrangement of the disks (the highlighted region in the figure). The underlying communication graph is shown with solid edges.}\label{fig_interf}
\end{figure}

\begin{definition}
Let $P=\{p_1,\ldots,p_n\}$ be a set of $n$ points in $\Re^d$ and let $r_1, \ldots,r_n$ be $n$ non-negative reals representing
power levels (or transmission radii) assigned to the points $p_1, \ldots, p_n$, respectively. Let $G=(P,E)$ be the graph associated with this power assignment, where $E = \{ \{p,q\} \colon  d(p,q) \leq min\{r_p,r_q\} \}$. That is, points $p,q$ are neighbors in $G$ if and only if $p$ is contained in the ball centered at $q$ with radius $r_q$ and {\em vice versa}. Let $ D = \{d_1,\ldots, d_n\}$ denote the set of balls where $d_i$ is the ball centered at $p_i$ and having radius $r_i$.
\end{definition}

Let $I(D)$ denote the maximum depth of the arrangement of the balls in $D$. That is $I(D) = \max_{q\in \Re^d} \{|\{d \in D \colon  q \in d \}|\}$. We call $I(D)$ the {\em interference} of $D$, which is also the {\em interference} of the network. Note that both $G$ and $I(D)$ are determined by $P$ and $r_1, \ldots, r_n$. Given a set $P$ of points in $\Re^d$, the {\em interference} of $P$ (denoted $I(P)$) is the smallest possible interference $I(D)$, where $D$ corresponds to a power assignment whose associated graph is connected.
The {\em interference minimization problem} asks for the power assignment for which $I(P)=I(D)$.

Empirically, (in dimension two) it has been observed that networks with high interference have high rates of message collision. This requires messages to be repeated often, which slows down the network and reduces battery life of the sensors \cite{rwz-2009-amiwasn}. Thus, a significant amount of research has focused on the creation of connected networks with low interference (see, e.g., \cite{HT-interference,k-mianbcr-12}). It is known that computing $I(P)$ (or even approximating it by a constant factor) is an NP-complete problem \cite{bbehm-imasn-14}, but some worst-case bounds are known.
\begin{theorem}[\cite{HT-interference}]
Let $P$ be a set of $n$ points in the plane. Then $I(P) = O(\sqrt{n})$. Furthermore, this bound is asymptotically tight, in the sense that for any $n$ there exists a set $P$ of $n$ points such that $I(P) = \Omega(\sqrt {n})$.
\end{theorem}

Here, we turn our attention to the kinetic version of the interference problem in arbitrary but fixed dimension.
We wish to maintain a connected graph of a set of moving points (representing moving sensors) that always has low interference. Unless the distances between sensors remain constant, no static radii assignment can work for a long period of time (since points will eventually be far from each other). Instead, we describe the network in a combinatorial way. That is, we look for a function $f:P\times [0,\infty) \rightarrow P$ that determines, for each sensor of $P$ and instance of time, its furthest away sensor that must be reached. Then, at time $t$ the communication radius of a sensor $p\in P$ is simply set equal to the distance between $p$ and $f(p,t)$. Ideally, we would like to construct a network that not only has small interference at any instance of time, but also the underlying graph has a small amount of combinatorial changes along time.

Our algorithm to maintain a connected graph is based on the ideas used in \cite{HT-interference} for the static case. We first pick a subset $N \subset P$ of ``hubs''. Those hubs will never change along time and will always have transmission radius big enough to cover all other points.
Each other point in $P\setminus N$ will be assigned at every instance of time to its nearest hub.
In the following we show that a careful choice of hubs will ensure a small interference, and overall small number of combinatorial changes in the radii assignment protocol. To bound the number of combinatorial changes, we need to use the machinery of Davenport-Schinzel sequences:
A finite sequence $\Sigma=(e_1, \ldots, e_m)$ over an alphabet of $n$ symbols is called a \emph{Davenport-Schinzel sequence} of order $t$  when no two consecutive elements of $\Sigma$ are equal, and for any two distinct symbols $x,y$, there does not exist a subsequence where $x$ and $y$ alternate $t + 2$ times. Several bounds are known on the maximum length of Davenport-Schinzel sequences of a given order. In particular, we are interested in upper bounds. See \cite{DSbook} for more details on Davenport-Schinzel sequences.

\begin{theorem}[Upper bound on Davenport-Schinzel sequences \cite{Nivasch}]
\label{the-DS}
A Davenport-Schinzel sequence of order $t$ on $n$ symbols has length at most $O(n2^{O(\alpha(n)^{\lfloor (t-2)/2 \rfloor})})$, where $\alpha(n)$ is the inverse of the Ackermann function.
\end{theorem}

The Ackermann function is a function that grows very rapidly, hence its inverse is usually regarded as a small constant (indeed, it is known that $\alpha(n)\leq 5$ for any input that can be stored explicitly in current computers). Davenport-Schinzel sequences are often used to bound the complexity of upper (or lower) envelopes of polynomial functions. Whenever we have a family of $n$ functions such that no two graphs of those functions cross more than $t$ times (for some bounded constant $t$), we can use Theorem~\ref{the-DS} to bound the complexity of their upper and lower envelope. 

\begin{theorem}
Let $P$ be a set of $n$ moving points in $\Re^d$ with bounded description complexity $s$. Then, there is a power assignment with updates, such that at any given time $t$ the interference of the network is at most $O(\sqrt{n \log n})$. Moreover, the total number of combinatorial changes in the network is at most $O^*(n^{1.5}\sqrt{\log n})$, where the $O^*$ notation hides a term involving the inverse Ackermann function that depends on $d$ and $s$.
\end{theorem}
\begin{proof}
We use Corollary~\ref{balanced-assignment} for some value of $k$ that will be determined later. We obtain a set $N$ of size $O(k\log k)$ with the properties guaranteed by Corollary~\ref{balanced-assignment}. The elements of $N$ are called {\em hubs}, and we assign to each of them the largest possible radius. That is, at any instance of time $t\geq 0$, a point $p \in N$ is assigned the distance to its furthest point in $P$. In other words, $f(p,t)$ is equal to the point $q \in P$ that maximizes the distance $d(p(t),q(t))$. Other points of $P$ are assigned the distance to their nearest hub. More formally,  $f(p,t)$, for a point $p \in P \setminus N$, is equal to the point $q\in N$ that minimizes the distance $d(p(t),q(t))$. Equivalently, if we consider the Voronoi diagram with sites $N(t)$, the function $f(p,t)$ will match $p(t)$ with the site associated to the Voronoi cell that contains $p(t)$ at time $t$.

First observe that the network is connected: indeed, all hubs are connected to each other forming a clique. Moreover, each point of $P\setminus N$ has radius large enough to reach one point of $N$. In particular, any two points of $P$ can reach each other after hopping through at most two intermediate sensors of $N$ (thus, the constructed network has diameter $3$).

We now pick the correct value of $k$ so that the interference of this protocol is minimized. Since $N$ has $O(k\log k)$ points, the overall interference contribution by hubs is bounded by the same amount. By Corollary~\ref{sparse-cell}, we also know that no point $q\in \Re^d$ can be reached by more than $O(n/k)$ points of $P\setminus N$ at any instance of time. That is, the total interference of any point $q\in \Re^d$ is at most $O(k\log k)$ from hubs, and at most $O(n/k)$ from non-hubs. Thus, by setting $k=\sqrt{n/\log n}$ we obtain the claimed bound.

We now bound the total number of combinatorial changes that will happen to the network along time. Let $p\in P$, we will show that the number of combinatorial changes of $p$ is bounded. Recall that, if $p$ is a hub it will connect to its furthest away point of $P$. Otherwise, $p$ will connect to its nearest hub. In either case, it suffices to bound the number of combinatorial changes of the nearest/furthest point within a group of moving points with respect to the moving point $p$. Equivalently, we are looking at the number of combinatorial changes of the upper envelope of the family of functions $\mathcal{F}_1=\{d(p(t),p'(t)) \colon p'\in P\}$ for points $p \in N$, or the lower envelope of the family of functions $\mathcal{F}_2=\{d(p(t),p'(t)) \colon p'\in N\}$ for points $p \not\in N$. By the bounded description complexity assumption, functions of $\mathcal{F}_1$ and $\mathcal{F}_2$ are such that the graphs of any pair of them cross $O(s)$ times. Thus, by the Davenport-Schinzel Theorem we can bound the number of combinatorial changes of the upper envelope of $\mathcal{F}_1$ by $O(\lambda_{O(s)}(n))$, where $\lambda_{t}(m)$ denotes the maximum length of a Davenport-Schinzel sequence of order $t$ on $m$ symbols. Similarly, the number of changes of the lower envelope of $\mathcal{F}_2$ is bounded by $O(\lambda_{O(s)}(\sqrt{n\log n}))$.

Ignoring the terms that depend on the inverse of the Ackermann function, we have that for any fixed constant $s$, $\lambda_{t}(m)=O^*(m)$. Combining this with the fact that we have $O(\sqrt{n\log n})$ hubs and at most $n$ non-hub points, the overall number of combinatorial changes is bounded by $O^*(n\times \sqrt{n\log n}+\sqrt{n\log n}\times n)=O^*(n^{1.5}\sqrt{\log n})$ as claimed.
\end{proof}

\section{Other Applications}\label{sec-other}

In this section we mention a few additional results that directly follow from Theorem~\ref{constant-description}. We hope that further research will reveal other interesting applications that stem from this or similar theorems.

\subsection{Approximate kinetic range counting}
Range counting is the problem of counting how many points are present in a given query range. More precisely, given a set $P$ of $n$ points in $\Re^d$ the goal is to preprocess these points so that given a query range $r$ (usually a halfspace, a sphere or some similar simple shape) we can determine the number of points in $r \cap P$. Exact range counting is difficult and the best results require superlinear memory or have query times polynomial in $n$~\cite{chazelle1989lower}. Consequently, more research has gone towards approximate range counting.

The problem of range counting can be approximated in several ways. First, one could base the approximation on the range. That is, we require that points that are far from the boundary of the query range are properly counted, but nothing is required of the points that lie close to the boundary.
That is, points that are close to the boundary of the query range may or may not be counted, but those clearly inside the range are guaranteed to be counted.
This form of approximation for the kinetic setting was considered by Abam~\etal~\cite{abs-kkdtlskdt-09}.

Another way to approximate range counting is simply by the number of reported points. When the number of points within the range is $k$ we wish to report a number $k'$ so that $(1-\eps) k \leq k' \leq (1+\eps) k$. It is difficult to certify such type of approximation, since ranges that contain few points must often report the exact number (in particular, we should be able to perform exact emptyness queries). To avoid this issue a common standard for approximate range counting is to use an $\eps$-approximation:

\begin{definition}\label{def:eps-approximation}
Let $(P,\R)$ be a hypergraph. A subset $A \subset P$ (not necessarily a hyperedge) is called $\eps$-approximation
if for any range $r \in \R$ the following holds:
$$
\left|\frac{\cardin{r \cap A}}{\cardin{A}}- \frac{\cardin{r \cap P}}{\cardin{P}}\right| \leq \eps.
$$
\end{definition}

In other words, $A$ is a sample of the points that represents the size of the hyperedges in the underlying hypergraph up to an absolute error $\eps$. It is straightforward to verify that every $\eps$-approximation is also an $\eps$-net, but the reverse does not always hold. In general, it is known that if $(P,\R)$ has $\VC$-dimension $d$ then a random sample of size $O(\frac{d}{\eps^2})$ is an $\eps$-approximation with at least some positive constant probability \cite{Tal94,LLS01}.

It is straightforward to apply this generalization of $\eps$-nets to obtain an approximation for range counting: construct an $\eps$-approximation $A$ of $P$ and construct an exact range counting structure on $A$. Then for each query range $r$ we perform the query on $A$ and multiply the result by $|P| / |A|$.

The following theorem follows immediately from Theorem~\ref{constant-description}. %
\begin{theorem}\label{theo_epsapprox}
Let $P$ be a set of moving points in $\Re^d$ with bounded description complexity and let $\cal R$ be a family of regions with bounded description complexity. Then for any $\eps \in (0,1]$ the kinetic hypergraph $(P, \cal R)$ admits an $\eps$-approximation of size $O(\frac{1}{\eps^2})$.
\end{theorem}

Notice that, as with $\eps$-nets, the set $A$ does not change throughout the motion. Using this result, we can apply the $\eps$-approximation-based approach for approximate range counting in the kinetic case.
We obtain the same running time as in the static case. 

\begin{corollary}
Let $P$ be a set of $n$ moving points in $\Re^d$ with bounded description complexity and let $\cal R$ be a family of regions with bounded description complexity. We can build an approximate range counting data structure using $\sigma$ space, for $m \leq \sigma \leq m^d$ and $m = \min(n, 1/\eps^2)$, in $O(n + m^{1+\delta} + \sigma (\log m)^\delta)$ time that answers queries in $O(\frac{m}{\sigma^{1/d}} \log^{d+1} \frac{\sigma}{m})$ time, for an arbitrarily small constant $\delta>0$. The relation between $k$, the reported number of points, and $\ell$, the real number of points in the range, is defined by
$\left|\frac{k - \ell}{n}\right| \leq \eps$.
\end{corollary}

\subsection{Discrepancy of kinetic range spaces}
Intuitively speaking, we say that a hypergraph $H=(V,\E)$ has small discrepancy if we can color its vertices with two colors, say `red' and `blue', such that the difference between
the red points and the blue points in every hyperedge is small.
A more formal definition is as follows: given a hypergraph $H=(V,\E)$, a {\em 2-coloring} of $H$ is a function $\chi: V \rightarrow \{-1,1\}$.
For a hyperedge $S \in \E$ let $\chi(S) = \sum_{v \in S} \chi(v)$, and $disc(H) = min_{\chi} max_{S \in \E} \cardin{\chi(S)}$.
We call $disc(H)$ the  {\em discrepancy} of $H$. In other words the discrepancy of $H$ is the difference between the number of red and blue points in the most imbalanced hyperedge in the `best' red-blue coloring possible for $H$.
The notion of discrepancy of a hypergraph is one of the deepest notions in combinatorics and has many applications. 

The following well known theorem provides a bound on the discrepancy of a hypergraph in terms of its shatter function; see, e.g., \cite{MATOUSEK-disc}.
\begin{theorem}[Primal shatter function bound]
\label{primal-shatter-bound}
Let $d > 1$ and $C$ be constants, and let $H=(V,\E)$ be a hypergraph on $n$ vertices with primal shatter
function satisfying $\pi_H(m) \leq  Cm^D$ for all $m \leq \cardin{V}$. Then
$disc(H) \leq C'n^{\frac{1}{2}-\frac{1}{2d}}$ where the constant $C'$ depends on $C$ and $D$.
\end{theorem}

The following theorem is an immediate consequence of Theorem~\ref{primal-shatter-bound} together with Theorem~\ref{kinetic-halfspaces}.
\begin{theorem}
Let $P$ be a set of $n$ moving points in $\Re^d$ with bounded description complexity $s$ and let ${\cal H}_d$ denote the family of all halfspaces. The kinetic hypergraph $H=(P,{\cal H}_d)$ satisfies $ disc(H) = O(n^{\frac{1}{2}-\frac{1}{2d+2}})$.
\end{theorem}

\section{Conclusion}\label{sec-conclusions}
Using the the machinery of \VC-dimension we have shown that the difference between static and kinetic environments for our facility location problem is small. We believe that a similar approach can be used for other problems.
We hope that future research will show other interesting applications of $\varepsilon$-nets in kinetic environments.

In Section~\ref{sec-applis-balanced-vor} we argued that it is unlikely that the ``balanced'' property can be significantly improved. Similarly, it seems unlikely that the ``reasonable'' constraint can be removed, even in one dimension.
Indeed, if points are allowed to move arbitrarily, they can create all $n!$ orderings along time. In particular, for any set $N\subset P$ we can always find a time and range that contains all points of $P\setminus N$. Thus, no subset $N \subset P$ can act as an $\eps$-net for all instances of time. Further note that, since the alternation in orderings can be repeated arbitrarily many times, the number of times that we need to change the set $N$ can also be unbounded. This behaviour can be created with trigonometric functions of low description complexity.

\section*{Acknowledgements}
This work was initiated during the Second Sendai Winter Workshop on Discrete and Computational Geometry. The authors would like to thank the other participants for interesting discussions during the workshop, as well as Alexandre Rok for helpful discussions.

\bibliographystyle{plain}
\bibliography{references}

\end{document}